\documentclass[letterpaper, 10 pt, conference]{ieeeconf}  

\IEEEoverridecommandlockouts                              

\usepackage{graphics} 
\usepackage{epsfig} 
\usepackage{graphicx}
\usepackage{xcolor}

\usepackage{times} 
\usepackage{amsmath} 
\usepackage{amssymb}  
 
\usepackage{amsthm}
\usepackage{mathtools} 

\theoremstyle{plain}
\newtheorem{theorem}{Theorem}
\newtheorem{proposition}{Proposition}

\theoremstyle{definition}
\newtheorem{definition}{Definition}
\newtheorem{assumption}{Assumption}

\usepackage{hyperref} 
\usepackage[capitalize]{cleveref}

\usepackage{subcaption}
\usepackage{graphicx}

\crefname{theorem}{theorem}{theorems}
\Crefname{theorem}{Theorem}{Theorems}

\crefname{proposition}{proposition}{propositions}
\Crefname{proposition}{Proposition}{Propositions}

\crefname{definition}{definition}{definitions}
\Crefname{definition}{Definition}{Definitions}

\crefname{assumption}{assumption}{assumptions}
\Crefname{assumption}{Assumption}{Assumptions}

\newcommand{\R}{\mathbb{R}}
\DeclareMathOperator{\col}{col}
\DeclareMathOperator{\blkdiag}{blkdiag}
\DeclareMathOperator{\NE}{NE}
\DeclareMathOperator{\VI}{VI}
\newcommand{\norm}[1]{\left\lVert #1 \right\rVert}
\newcommand{\opt}{\star}
\newcommand{\U}{\mathcal{U}}
\newcommand{\Z}{\mathcal{Z}}
\newcommand{\NC}{\mathcal{N}}
\newcommand{\Eq}{\mathcal{E}}

\newcommand{\Gj}[1]{G^{(#1)}}
\newcommand{\Fj}[1]{F^{(#1)}}
\newcommand{\Fu}[1]{F^{(#1)}_{u}}
\newcommand{\Fx}[1]{F^{(#1)}_{x}}
\newcommand{\Sj}[1]{\mathcal{S}^{(#1)}}
\newcommand{\uj}[1]{u^{(#1)}}
\newcommand{\phij}[1]{\phi^{(#1)}}

\newcommand{\xinit}{\mathbf{x}}
\newcommand{\ureal}{u^{\circ}}
\newcommand{\Bhat}{\hat{B}}
\newcommand{\xbar}{\bar{x}}

\newcommand{\At}{\tilde{A}}
\newcommand{\Bt}[1]{\tilde{B}_{#1}}
\newcommand{\Qt}[2]{\tilde{Q}_{#1}^{(#2)}}

\newcommand{\qt}[2]{\tilde{q}_{#1}^{(#2)}}
\newcommand{\monc}{\mu}

\hypersetup{
  pdftitle={Stability Analysis for Objective Misspecifications Among Model Predictive Game Controllers},
  pdfauthor={Ada Yildirim},
  pdfkeywords={model predictive control, game theory, MPC, dynamic games},
  pdfsubject={Control Systems}
}

\title{\LARGE \bf Stability Despite Misalignment:\\ Model Predictive Game Controllers with Objective Misspecifications}

\author{Ada Yıldırım, Brandon Collins and Bryce L. Ferguson%
\thanks{A. Yıldırım, B. Collins, and B. L. Ferguson are with the Thayer School of Engineering, Dartmouth College, Hanover, NH 03755, USA,
        {\tt\small \{ada.yildirim.th, brandon.c.collins, bryce.l.ferguson\}@dartmouth.edu}}%
}

\begin{document}

\bstctlcite{BSTcontrol}

\maketitle
\thispagestyle{empty}
\pagestyle{empty}

\begin{abstract}

Predictive control requires the controller to possess an embedded prediction model. In multi-agent control, this prediction model must consider the behavior that emerges from other agents' strategic decisions. Game theory offers a natural model for the emergent collective behavior of self-interested agents.
Model predictive games (MPGs) are a class of controllers in which an agent iteratively solves a finite-horizon game to predict the behavior of a multi-agent system and synthesize their own control action. When all agents use the same game model to select their control action, conditions have been found that guarantee the closed-loop stability of the multi-agent system. However, when multiple agents implement their own MPG controllers, their respective game models may be misspecified due to inaccurate estimates or conjectures about other agents' objectives. This paper analyzes the consequences of these prediction misalignments on the system's behavior. We show that existing stability conditions for homogeneous MPG controllers do not extend to the heterogeneous setting, even when they hold for each agent's individual game. We provide sufficient conditions for the stability of multi-agent dynamical systems with heterogeneous MPG controllers adopting linear-quadratic (LQ) games. We also show that misaligned game models can destabilize a system that reaches stability when all agents share the same game model, using a financial duopoly simulation.

\end{abstract}

\section{INTRODUCTION}

Multi-agent control has been increasingly employed for complex real-life systems, such as multi-vehicle autonomous driving \cite{zhang_efficient_2024}, multi-drone racing \cite{papuc_strategizing_2026}, and smart grids with distributed energy generation and storage \cite{atzeni_demand-side_2013}. For predictive control in multi-agent dynamical systems, other agents' behavior needs to be modeled for accurate prediction and planning. Interactions of self-interested and self-determining agents are often approximated as noncooperative dynamic games \cite{basar_dynamic_1999}. The role of solution concepts in games, e.g., Nash equilibria (NE), is intended to provide predictions of reasonable or plausible behavior among these agents. This notion of prediction naturally aligns with the objective of multi-agent control design. However, the lack of accuracy of the game-theoretic solutions as collective behavior predictors raises the question of the reliability of such approaches. 

\begin{figure}
    \centering
    \includegraphics[width=0.75\linewidth]{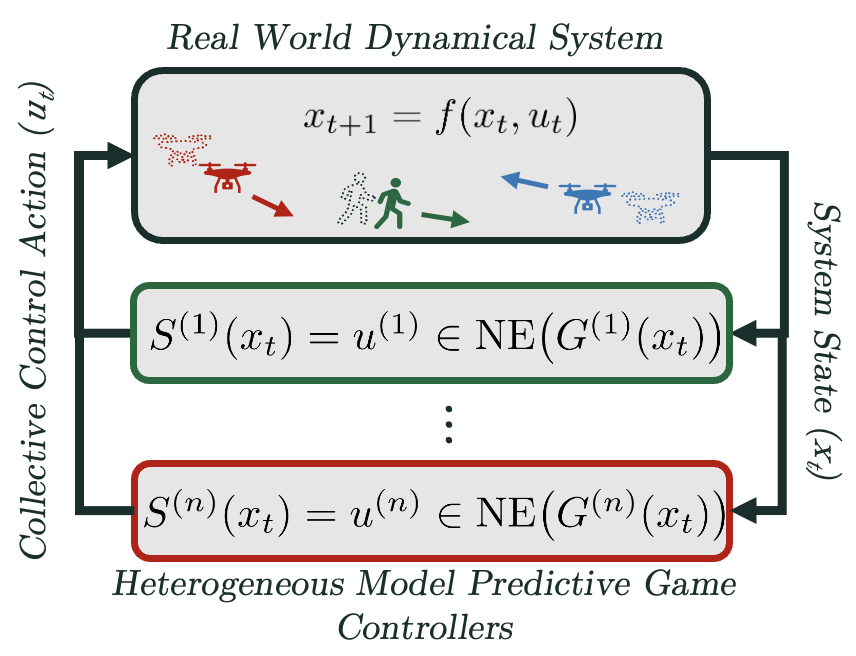}
    \caption{
    Multi-agent dynamical system with heterogeneous MPG controllers. At each time step, each agent $j$ solves its own conjectured game $G^{(j)}(x_t)$, and applies the first action of its NE prediction. Together, these actions form the collective control action $u_t$.}
    \label{fig:mpg}
\end{figure}

Recent works have investigated the connection between games and control and the resulting behavior of these interactions. In multi-agent reinforcement learning~\cite{zhang_multi-agent_2021}, agents learn equilibrium strategies with data-driven interactions. However, this approach typically assumes agents can interact over long durations and may not sufficiently inform ex ante control design. Open-loop and closed-loop NE in differential games~\cite{basar_dynamic_1999} provide a complementary perspective, where control policies are found at NE. However, most approaches assume accurate knowledge of the game model by all agents. 

Another line of research on inverse learning in games seeks to estimate the objectives of other agents through online interactions~\cite{ward_active_2024,soltanian2025pace,khan_what_2025}. Either by inaccurate approximation, insufficient data in estimation, or the need to design offline rather than adapt online, some level of misspecification between the agents' conjectured game models and realized behavior persists. The resulting loss in performance has been formalized as ``game-to-real gap" \cite{ferguson_game--real_2026}.

In the single-agent setting, model predictive control (MPC) is a common tool for choosing a control action at the current time by solving a finite-horizon open-loop optimal control problem using a dynamic model.
This solution forecasts system behavior over a finite horizon and closes the loop via iteratively repeating this control action synthesis \cite{rawlings_model_2020}.
In the multi-agent setting, similar philosophies of control design are emerging. The idea of iterative game-theoretic planning for multiple decision-making agents has been studied in \cite{fridovich-keil_efficient_2020}. 
Formally, receding horizon games replace the open-loop finite-horizon optimal control solution of MPC with an open-loop, finite-horizon NE solution.
Recent works have studied the dynamics of receding horizon games and established their stability with linear-quadratic (LQ) game models \cite{hall_stability_2025, benenati_linear-quadratic_2025}.

When an agent utilizes a receding horizon game to select their control action, the resulting controller is called a model predictive game (MPG) controller~\cite{papuc_strategizing_2026}.
Specifically, at each time step, an agent using an MPG controller solves for an open-loop finite-horizon NE within some embedded game model. From this solution, they select their individual action and deploy it in the system.
The emerging control architecture of MPGs~\cite{benenati_linear-quadratic_2025,hall_multiparametric_2025} allows agents to adapt to realized behavior, but its stability analysis has assumed the equilibrium is solved and deployed homogeneously across agents~\cite{hall_stability_2025}.
However, when controllers for agents are designed separately, heterogeneity occurs among agents' game models, and the accuracy of the game model is dependent on their knowledge or conjecture of the other agents. Competitive agents do not have full information on other agents’ objectives, which brings the misspecification into their conjectured game model, e.g., in receding horizon game models of competitive supply chains, each firm plans with its own estimates of its competitors' parameters~\cite{hall_receding_2024}.
Current stability certificates do not cover the case of heterogeneity among conjectured game models due to misspecifications.

In this work, we consider a dynamical system in feedback with heterogeneous MPG controllers, differing through the objective misspecifications in their game models, as illustrated in Fig.~\ref{fig:mpg}. Our main contributions are threefold: 1) we show that even if each conjectured game can individually yield a stable MPG controller had it been adopted by all agents, combining the agents’ selected actions from heterogeneous game models can destabilize the system, 2) we provide a verifiable linear matrix inequality (LMI) that certifies the stability of a multi-agent system with heterogeneous MPG controllers with LQ game models, and 3) we demonstrate and analyze the consequences of increased misspecifications in a financial duopoly. 

\section{MODEL SETUP}

\subsection{Continuous Action Monotone Games}\label{subsec:monotone_games}

A mathematical game formulation consists of agents, their actions, and cost functions: $G = (N, \{\U_i\}_{i \in N}, \{J_i\}_{i \in N})$, where $N = \{1,2,\ldots,n\}$ is the set of players, $\U_i \subseteq \R^{m_i}$ is the action set of player $i \in N$, with $m_i \in \mathbb{N}$ denoting the dimension of player $i$'s action, $\Z = \U_1 \times \cdots \times \U_n$ is the joint action space, and $J_i : \Z \to \R$ is the cost function of player $i \in N$.
The collective behavior of the group is captured by the joint action profile consisting of the individual actions from each agent, which is represented by \( u=(u_1,\dots,u_n) \in \R^m\), where $m = \sum_{i \in N} m_i$.
We denote the action of all players except $i$ as $u_{-i}$.

In game-theoretic contexts, the emergent behavior of the agents is often modeled by an NE.
Within the context of multi-agent planning, this solution concept can serve as a prediction of the collective behavior.
Throughout this work, we will consider only local constraints on the agents' actions and the consequent NE game solution concept.
\begin{definition}\label{def:nash}(Nash equilibrium): Consider an $n$-player game with a joint action space $\Z = \U_1 \times \U_2 \times \cdots \times \U_n$  and cost functions
$J_i : \Z \to \R$. The joint action profile $u^\opt=(u_1^\opt,\dots,u_n^\opt)$ is a Nash equilibrium if the following holds for every player $i$:
\begin{equation}\label{eq:Nash_def}
J_i(u_i^\opt, u_{-i}^\opt) \leq J_i(u_i, u_{-i}^\opt) 
\quad \forall u_i \in \U_i, \; \forall i \in N.
\end{equation}
\end{definition}
In general, pure NE may not exist for a game, or there may exist multiple within a single game \cite{fudenberg_game_1991}, making their role as predictors unreliable. As such, many works have studied the class of strongly monotone games \cite{hall_stability_2025, ferguson_game--real_2026}. This class of games consists of non-cooperative agents who choose their strategies from a convex set of actions, and the pseudo-gradient of each agent's cost, defined as \(F(u)= \big( \nabla_{u_i} J_i(u_i, u_{-i}) \big)_{i \in N} \), is strongly monotone \cite{facchinei_finite-dimensional_2004}, resulting in a unique NE.
\begin{definition}\label{def:monotonicity}(Strong monotonicity): An operator
\( F: \mathcal{X} \subseteq \R^{m} \to \R^{m} \) is strongly monotone on \( \mathcal{X} \) if there exists \( \monc > 0 \) such that,
\[
\bigl( F(v_1) - F(v_2) \bigr)^{\top} (v_1 - v_2)
\;\ge\; \monc \, \norm{v_1 - v_2}^2,
\]
for all \( v_1, v_2 \in \mathcal{X} \), where \( \monc \) is the strong monotonicity constant.
\end{definition}
NE of strongly monotone games where $\Z$ is convex can be equivalently defined in the form of a variational inequality (VI) problem of the pseudo-gradient \(F\) over \(\Z\).
\begin{definition}[] \label{def:vi}
(Variational inequality): Given a subset \(\Z \) of  \(\R^{m}\) and a mapping $F: \Z \rightarrow \R^{m}$, a vector $u \in  \Z$ solves the VI problem, denoted \(\VI(\Z,F)\), if
\begin{equation}\label{eq:VI_def}
    (y-u)^TF(u) \geq 0, \forall y \in \Z.
\end{equation} 
The set of solutions to this problem is denoted $\mathcal{S}(\Z,F)$.
\end{definition}
If each $J_i$ is convex in $u_i$ and continuously differentiable over $\Z$, then the solution set $\mathcal{S}(\Z,F)$ is equivalent to the set of NE of the game as given below.
\begin{proposition}[Facchinei and Pang 2004~\cite{facchinei_finite-dimensional_2004}] \label{prop:nash_VI}
Let each $\U_i$ be a closed convex subset of \(\R^{m_i}\) and $F\equiv(\nabla_{u_i} J_i(u_i,u_{-i}))_{i \in N}$.  Suppose that for each fixed tuple $u_{-i}$, the function $J_i(u_i,u_{-i})$ is convex and continuously differentiable in $u_i$. Then a tuple $u \in \Z$ is an NE if and only if $u \in \mathcal{S}(\Z,F)$.
\end{proposition}
With strong monotonicity, the solution to \eqref{eq:VI_def}, and thus the NE of the game, is unique.
\begin{proposition}[Bauschke and Combettes 2011~\cite{bauschke_convex_2017}]\label{prop:singleton}
If $F$ is Lipschitz continuous, strongly monotone, and $\Z$ is closed and convex, then the solution mapping $ \mathcal{S}(\Z,F)$ is a singleton. 
\end{proposition}
Strong monotonicity holds for several classes of games, including LQ games under suitable conditions on the cost weights~\cite{benenati_linear-quadratic_2025}, which is the form of the games adopted in this paper. In addition, they can be used to approximate more complex multi-agent interactions~\cite{fridovich-keil_efficient_2020}, particularly in online or real-time implementations. A controller solving for a game model in a receding-horizon fashion can be designed in the form of an MPG controller.

\subsection{Model Predictive Games}\label{subsec:mpg}

In dynamic multi-agent systems, MPG controllers have emerged as a promising archetype that adapts to the behavior of other agents while retaining strategic planning capabilities. Like their namesake, model predictive controllers, MPG controllers generate a finite-horizon prediction to synthesize the current control action. They differ in that, rather than solving an open-loop optimal control problem, MPG solves for an open-loop NE of some embedded game model.

We consider a dynamic multi-agent environment in which the system state $x \in \R^{d}$, shared by all agents, where $d$ denotes the dimension of the global state, evolves according to the linear time-invariant (LTI) dynamics,
\begin{equation} \label{eq:dynamics}
        x_{t+1} = A x_t + \sum_{i \in N} B_i u_{i,t}, 
\end{equation}
where $u_{i,t}$ is the control action of agent $i \in N$ at time $t$.

We now consider a feedback controller of the form $u_{t} = \kappa(x_t)$, which provides the next control action. An MPG controller solves for an open-loop, finite-horizon NE given the LTI dynamics and the stage costs of agents. Adopting the LQ game model, we assume that each stage cost is of the following form:
\begin{equation} g_{i}(u_{i,t}, u_{-i,t},x_t) = x_t^\top Q_i x_t + x_t^\top q_i + u_t^\top R_i u_t, \label{eq:stage_cost} \end{equation}
for all $i \in N$ where $Q_i = Q_i^\top \succeq 0$ and $R_i^\top=R_i$.

At time $t$ and state $x_t = \xinit$, the finite-horizon game is formulated as, following the setup of~\cite{hall_stability_2025,benenati_linear-quadratic_2025}:
\begin{align}
    &\begin{array}{rl}
        i \in N: & 
        \left\{
        \begin{aligned}
            \min_{u_i} \quad & \sum_{k=0}^{K-1} g_i(u_{i,k}, u_{-i,k}, x_k) \\
            \text{s.t.} \quad & x_{k+1} = A x_k + \sum_{i \in N} B_i u_{i,k}, \\
            & u_{i,k} \in \U_i,\quad k \in \{0,\ldots,K-1\} \\
            & x_0 = \xinit
        \end{aligned}
        \right.
    \end{array}\nonumber\\[-15pt]
    &~\label{eq:fh_game}
\end{align}
where \(K>1\) is the prediction horizon\footnote{The finite-horizon MPC approach approximates a solution to the infinite-horizon version, which can be computationally intractable~\cite{rawlings_model_2020}. Here, the finite-horizon game offers tractable game-theoretic solutions.}. From \Cref{prop:nash_VI}, this joint-control signal is the solution to a VI. As there are no constraints on the state, \eqref{eq:fh_game} is feasible for all $\xinit\in\R^{d}$, and the solution ensures satisfaction of local constraints on the inputs.

Substituting the LTI dynamics and initial condition into the stage costs, the compact formulation of the game becomes:
\begin{equation}
    \forall i \in N: \quad
    \left\{
    \begin{array}{cl}
        \displaystyle \min_{u_i} & J_i(u,\mathbf{x}) \\[2pt]
        \text{s.t.} & u \in \mathcal{Z}
    \end{array}
    \right.
    \label{eq:compact_game}
\end{equation}
where $\Z=\U_1^K\times\cdots\times\U_n^K$.
The mapping from the initial state \(\xinit\) to the solution of the VI \eqref{eq:VI_def} or, equivalently, to the NE of \eqref{eq:compact_game} is $\mathcal{S}(\Z,F(\cdot;\xinit))$ which we define concisely as:
\begin{equation}
    S(\xinit) := \left\{ u \;\middle|\; (y-u)^TF(u;\xinit) \geq 0, \forall y \in \Z \right\},
    \label{eq:solution_map_general}
\end{equation}
where $F$ is the pseudo-gradient of the game in \eqref{eq:compact_game}. $F$ can be written as the sum of two separate terms dependent only on \(u\) and \( \xinit \), respectively, as
$
    F(u; \xinit)=F_u(u)+F_x\xinit.
$

At each time step \(t\), from the finite horizon NE prediction $u$, the instantaneous control actions are selected, i.e., $\ureal_0= \Xi S(\xinit)$, where $\Xi$ is a selection matrix. With the instantaneous control actions implemented in feedback, the closed-loop system becomes
\begin{equation}\label{eq:closed_loop_system_mpg}
    x_{t+1} =  A x_t + B\kappa(x_t) = A x_t +B\Xi S(x_t),
\end{equation}
where $B=[B_1,\dots,B_n]$, and $\kappa$ is the feedback law for the overall dynamical system.

Despite their demonstrated success in homogeneous settings~\cite{hall_stability_2025,benenati_linear-quadratic_2025,hall_multiparametric_2025}, MPG controllers are often designed in isolation in competitive or distributed settings, and the agents do not solve a shared game model. This requires the designer to conjecture or estimate the objectives of the agents they do not control. Differences between a designer's conjectured or estimated game model and the true objectives of other agents can cause a gap between the intended behavior of the controller and its performance in the real world.
The following subsection formalizes how these misspecifications appear in MPG controllers.

\subsection{MPG Controllers with Misspecifications}

We consider that player $j \in N$ synthesizes its control action from a conjectured game model.
Formally, agent $j$ conjectures that the objective of agent $i$ is characterized by the cost function \( J_i^{(j)} \), where throughout the notation \((\cdot)^{(j)}\) denotes some quantity that is conjectured or deduced by agent \(j\).
With these conjectures, agent $j$ possesses the game model $\Gj{j} = (N, \{\U_i\}_{i \in N}, \{J_i^{(j)}\}_{i \in N})$. $\Gj{j}$ is parameterized by the initial state $x_t = \xinit$ and used to predict the collective behavior of the group by solving for an NE solution $\uj{j} \in \NE(\Gj{j})$.
The agent then synthesizes their control action from their NE solution, i.e., they use the action $u_j^{(j)}$.

\begin{figure*}
\centering
\begin{subfigure}{0.32\textwidth}
    \includegraphics[scale=0.95]{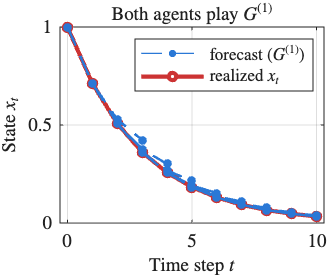}
\end{subfigure}
\hfill
\begin{subfigure}{0.32\textwidth}
    \includegraphics[scale=0.95]{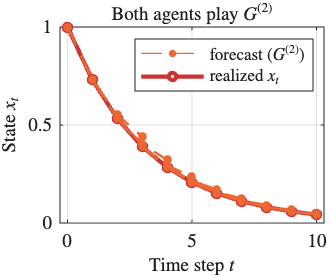}
\end{subfigure}
\hfill
\begin{subfigure}{0.32\textwidth}
    \includegraphics[scale=0.95]{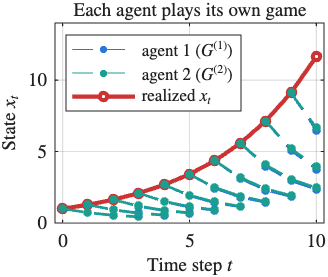}
\end{subfigure}
\vspace{-5pt}
\caption{State trajectories with initial state $x_0 = 1$. Dashed lines represent each agent's Nash forecast from the measured state; solid lines are the realized state. Left and middle plots show the case when both agents play the same game. Forecasts closely track the realized state, and the state decays. The right plot shows the case when each agent plays its own game. Both forecasts decay while the state grows, and the origin is unstable.}
\label{fig:per_game}
\end{figure*}

In this work, we are interested in the case where agents' conjectures are inaccurate, resulting in heterogeneous game models, or $\Gj{i} \neq \Gj{j}$.
In this setting, agents' predictions can in general be misaligned from one another, i.e., $\uj{i} \neq \uj{j}$.
When each agent selects their action from their local prediction, the \emph{realized joint action} will be $\ureal:=\left(u_1^{(1)},\ldots ,\ u_n^{(n)} \right)$.
Note that despite each agent solving for an NE, $\ureal$ need not be an equilibrium of any individual player's conjectured game. 

Within the closed-loop multi-agent dynamical system setting, this translates to having individual MPG controllers per agent in the feedback. We adopt the LTI dynamics in (\ref{eq:dynamics}) for the dynamical system block. 
Specifically, we consider that each agent uses an MPG controller in (\ref{eq:fh_game}) designed with their game model $G^{(j)}$. Specifically, the stage cost functions $ g_{i}$ are replaced with the conjectures of the other agents' stage costs $ g_{i}^{(j)}$ in the form of (\ref{eq:stage_cost})
with $Q_i^{(j)} \succeq 0$, $q_i^{(j)}$ and $R_i^{(j)}$ for all $i,j \in N$.

It is assumed that all agents have full information on the system dynamics, thus, only objective misspecification on the stage cost $g_i^{(j)}$, therefore, the overall cost function $J_i^{(j)}$ is present in the conjectured game model. The controller for each agent $i \in N$ is of the form $u_{i,t} = \kappa_i(x_t)$, which provides the next control action for the agent.

The closed-loop system of each agent deploying an MPG controller, with heterogeneous conjectured games, can be formalized as follows:

\noindent\textit{1) Each agent solves for open-loop, finite-horizon NE:} At time $t$ and state $x_t = \xinit$, agent $j$ formulates the finite-horizon game in \eqref{eq:fh_game}.
Observe that the open-loop finite horizon game in \eqref{eq:fh_game} fits the definition of a continuous action game defined in \cref{subsec:monotone_games} where the control signal $u_{i,0:K-1}$ can be cast as a vector action constrained to $(\U_i)^K$, inducing the joint action space $\Z=\U_1^K\times\cdots\times\U_n^K$, and the player cost functions $J_i^{(j)}(\cdot;\xinit)$ are parameterized by the initial condition $\xinit$.
Therefore, from the point of view of the agent \(j\), in each time step, they solve for an NE of the parameterized game $\Gj{j}\left(\xinit\right)$, i.e., finding $\uj{j} \in \NE\left(\Gj{j}(\xinit)\right)$ of the form \(\uj{j}= \col(\uj{j}_{i,k})_{i \in N,k \in [K]} \).

\noindent\textit{2) Select instantaneous control action from NE:}
The solution mapping from the point of view of agent \(j\) is in the form of (\ref{eq:solution_map_general}) with $F^{(j)}(u;\xinit)$.
In Section \ref{sec:stability}, we apply assumptions so that $\Sj{j}$ is always well-defined and single-valued, in line with \Cref{prop:singleton}. 
At each time step \(t\), from the finite horizon NE prediction $\uj{j}$, agent $j$ selects their instantaneous control action as the first element of their individual control signal, i.e., $u_{j,0}^{(j)} = \Xi_j \Sj{j}(\xinit)$, where $\Xi_j$ is a selection matrix.

\noindent\textit{3) Closed-loop dynamics of heterogeneous MPG controllers:}
When each agent $j \in N$ utilizes an MPG controller and implements the action \(\uj{j}_{j,0}\), the realized joint action at time $t$ and state $x_t$ is
$\ureal_{t} = \kappa(x_t) :=  \left[(\Xi_1 \Sj{1}(x_t))^\top, \ldots, (\Xi_n \Sj{n}(x_t))^\top \right]^\top$.
With this state feedback, the closed-loop system becomes
\begin{equation}\label{eq:closed_loop_system}
    x_{t+1} =  A x_t + B\kappa(x_t) = A x_t + \sum_{j \in N} {B}_j\Xi_j \Sj{j}(x_t),
\end{equation}
where $\kappa$ is the feedback law for the overall dynamical system.

Existing stability results for MPG controllers consider a homogeneous game model shared by all agents, simplifying our closed-loop system \eqref{eq:closed_loop_system} to contain only one NE solution~\cite{hall_stability_2025}.
In practice, these controllers are deployed on competing or distributed agents, which conjecture their game models without full knowledge of others, resulting in misaligned models. Heterogeneous MPGs of this kind have been simulated for competitive supply chains without a stability analysis~\cite{hall_receding_2024}. 
To the best of the authors' knowledge, this is the first work to provide a rigorous stability analysis of multi-agent systems with heterogeneous MPG controllers differing in their game models.
We demonstrate in \Cref{subsec:motivating} that the extension of the existing certificates to the heterogeneous setting is non-trivial.
As such, we analyze the stability of systems with heterogeneous MPG controllers defined by the closed-loop system in \eqref{eq:closed_loop_system}.

\section{STABILITY ANALYSIS}\label{sec:stability}

\subsection{Motivating Example}\label{subsec:motivating}

Consider the scalar system $x_{t+1} = 0.8\,x_t + u_{1,t} + u_{2,t}$ with two agents, input sets $\mathcal U_i = [-50, 50]$ and horizon $K = 3$. The open-loop system is stable. If both agents applied $u_{i,t} = 0$, the state would converge to the origin at rate $0.8$. Agent~$i$ plans with the LQ game $G^{(i)}$ in Appendix~\ref{app:counterexample}, with its own true cost but a wrong conjecture of the other agent's cost. Both games satisfy the stability condition of~\cite{hall_stability_2025}, so if both agents played $G^{(1)}$, or both played $G^{(2)}$, the closed loop would have a globally asymptotically stable equilibrium point as seen in the left and center plots of Fig.~\ref{fig:per_game}. However, when each agent plays its own game, the origin is an unstable equilibrium of~\eqref{eq:closed_loop_system}, which now has three equilibria: the unstable origin and $x \approx \pm 309.0$, each of which attracts only the trajectories that start on its side of the origin. Therefore,~\eqref{eq:closed_loop_system} has no globally asymptotically stable equilibrium, and the long-run behavior depends on the initial state. The state settles far from the origin when the agents' predictions of each other's actions reach the input bounds. These limiting states grow with the bounds of the input sets. For $\mathcal U_i = [-\bar u, \bar u]$ they are $x \approx \pm 6.18\,\bar u$.
At every step, both agents forecast a decaying state, while the realized state grows as seen on the right plot of Fig.~\ref{fig:per_game}.

In this example, if both agents use $G^{(1)}$ in their MPG controller or if both agents use $G^{(2)}$, the stability certificate of~\cite{hall_stability_2025} is satisfied.
This example illustrates that these existing guarantees do not apply once the agents' game models differ. The growing set of results for homogeneous MPG controllers~\cite{hall_stability_2025, benenati_linear-quadratic_2025,hall_multiparametric_2025} therefore does not apply directly to the heterogeneous setting in~\eqref{eq:closed_loop_system}.
This observed fragility of existing MPG controllers to heterogeneity underscores the need for a certificate that accounts for misalignment among agents' conjectured game models, which can emerge in deployment.

\subsection{Stability Conditions}\label{subsec:stability_conditions}
In order to derive the conditions for closed-loop stability of the multi-agent dynamical system with heterogeneous MPG controllers \eqref{eq:closed_loop_system}, we make the following assumptions.

\begin{assumption} \label{as:mpc}
The following hold for every MPG controller as described in \eqref{eq:fh_game} designed with $G^{(j)}$ and $ g_{i}^{(j)}$, and used in feedback with the dynamical system:
\begin{itemize}
    \item[(i)] \textit{The open-loop dynamics \eqref{eq:dynamics} are stable, i.e., the spectral radius \( \rho(A) < 1 \).}
    \item[(ii)] \textit{The local sets \( \U_i, \forall i \in N, \)
    are closed and convex. The set \( \Z \) is compact and non-empty.
    \item[(iii)] The function \( g_{i}^{(j)}(\cdot, u_{-i}, \xinit) \) 
    is strongly convex and continuously differentiable for any fixed \( u_{-i} \), \( \xinit \), \( \forall i,j \in N \).}
    \item[(iv)] \textit{The pseudo-gradient \( \Fj{j}(\cdot; \xinit) \) of the game $G^{(j)}$ is \( \monc_j \)-strongly monotone for any fixed \( \xinit \), for some \(\monc_j > 0\), $\forall$ j $\in$ N.}

\end{itemize}
\end{assumption} 

These ensure that each conjectured game $\Gj{j}(\xinit)$ has a unique NE for every $\xinit \in \R^{d}$, i.e., $\Sj{j}(\xinit)$ is a singleton by \Cref{prop:singleton}. A strong monotonicity assumption is common in game theory literature to achieve results using the monotone operator theory~\cite{hall_stability_2025,belgioioso_distributed_2022,parise_variational_2019}. As in~\cite{hall_stability_2025}, we focus on open-loop stable systems, which can already lose their globally asymptotically stable equilibrium under misaligned game models (Section~\ref{subsec:motivating}). The following \Cref{thm:stability} provides a stability condition for a multi-agent dynamical system consisting of MPG controllers with objective misspecifications.

\begin{theorem}
\label{thm:stability}
Let \eqref{eq:closed_loop_system} define a multi-agent dynamical system with heterogeneous MPG controllers in the feedback, satisfying \Cref{as:mpc}.
If there exist a positive-definite matrix $P \succ 0$ and multipliers $\lambda_j>0$, $j\in N$,
such that
\begin{equation}
\label{eq:stability_lmi}
M(P) + \sum_{j\in N}\lambda_j W_j
\preceq -\varepsilon I,
\end{equation}
for some $\varepsilon>0$, where
\begin{equation}
\label{eq:M_def}
M(P)
:=
\begin{bmatrix}
A^\top P A - P & A^\top P \Bhat \\
\Bhat^\top P A & \Bhat^\top P \Bhat
\end{bmatrix},
\end{equation}
$\Bhat =\begin{bmatrix} {B}_1\Xi_1 & {B}_2\Xi_2 & \dots & {B}_n\Xi_n \end{bmatrix}$, $I$ is the identity matrix, and the blocks of each symmetric matrix $W_j$ are defined, for each $j\in N$, as:
\begin{subequations}\label{eq:W_def}
\begin{align}
[W_j]_{1,j+1} = -\frac{(\Fx{j})^\top}{2}, \ &[W_j]_{j+1,1} = -\frac{\Fx{j}}{2}, \\
[W_j]_{j+1,j+1} &= -\monc_jI,
\end{align}
\end{subequations}
and zero otherwise, where the subscript indices denote the block positions for $W_j$, $
F^{(j)}(u; \mathbf{x})=F^{(j)}_{u}(u)+F^{(j)}_{x}\mathbf{x}$, and $\monc_j$ is the strong monotonicity constant of \(\Fu{j}\), then there exists a globally asymptotically stable equilibrium point $\xbar\in\R^{d}$ of the closed-loop system~\eqref{eq:closed_loop_system}.
\end{theorem}

The proof is given in Appendix~\ref{app:proof}. \Cref{thm:stability} extends the dissipativity argument of~\cite{hall_stability_2025}, which assumes that all agents solve the same game, so that a single solution mapping and a single cocoercivity inequality close the loop. After observing that \cite{hall_stability_2025} is insufficient to characterize the stability of heterogeneous MPG controllers, we extend the formulation to the heterogeneous setting that agents may face at deployment. Here, each agent's solution mapping $\phij{j}$ is a separate block with its own cocoercivity inequality and multiplier $\lambda_j$. To the best of our knowledge, \Cref{thm:stability} is the first closed-loop stability certificate for MPG controllers with heterogeneous conjectures, which returns a single LMI of size $d + nKm$ that covers all controllers jointly. As in~\cite{hall_stability_2025}, the condition is sufficient but not necessary. 
When all agents use the same game model, our condition implies, but does not reduce to, the stability certificate of~\cite{hall_stability_2025}, since~\Cref{thm:stability} treats the agents' predictions as separate blocks even when they coincide. In return, it covers the heterogeneous setting, where certifying each game separately does not imply stability.

Condition~\eqref{eq:stability_lmi} is an LMI in $(P, \lambda_j, \varepsilon)$, so it can be checked with a semidefinite programming solver (e.g., YALMIP~\cite{lofberg2004yalmip} with MOSEK~\cite{mosek}) by maximizing $\varepsilon$ subject to~\eqref{eq:stability_lmi}, $P \succeq 0$, $\lambda_j \geq 0$ and $\operatorname{tr}(P) = 1$, where the last constraint only fixes the scale of $P$, since~\eqref{eq:stability_lmi} is homogeneous in $(P, \lambda_j, \varepsilon)$. \Cref{thm:stability} applies if the optimal $\varepsilon$ is positive, and the strict conditions $P \succ 0$ and $\lambda_j > 0$ then hold automatically.

In~\eqref{eq:stability_lmi}, $M(P)$ describes the change of the storage function $V_{\xbar}^0(x) = (x-\xbar)^\top P (x-\xbar)$ of the LTI block, and each $W_j$ encodes agent~$j$'s MPG controller through how strongly the state enters its game ($\Fx{j}$) and how strongly its prediction can respond to it, which is at most a factor $1/\monc_j$. Misspecifications enter only through these quantities, and unlike checking the stability condition separately for each $\Gj{j}$ (Section~\ref{subsec:motivating}), \eqref{eq:stability_lmi} contains all MPG controllers in a single inequality, each weighted by its own multiplier $\lambda_j$.
When \eqref{eq:stability_lmi} holds, $V_{\xbar} = V_{\xbar}^0 + \sum_{j\in N}\lambda_j V_{\xbar}^j = V_{\xbar}^0$ serves as a Lyapunov function for the closed loop \eqref{eq:closed_loop_system}.
Even in the presence of misspecifications within every agent's game model, closed-loop stability can be achieved, as demonstrated with a duopoly simulation in the next subsection.

\subsection{Duopoly Simulation}
\begin{figure*}[t]
  \centering
  \begin{minipage}[b]{0.32\textwidth}
    \centering
    \includegraphics[width=\linewidth]{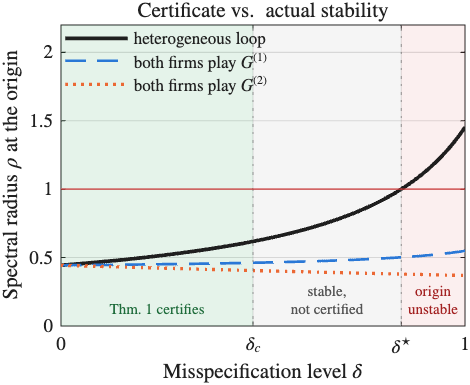}
  \end{minipage}\hfill
  \begin{minipage}[b]{0.32\textwidth}
    \centering
    \includegraphics[width=\linewidth]{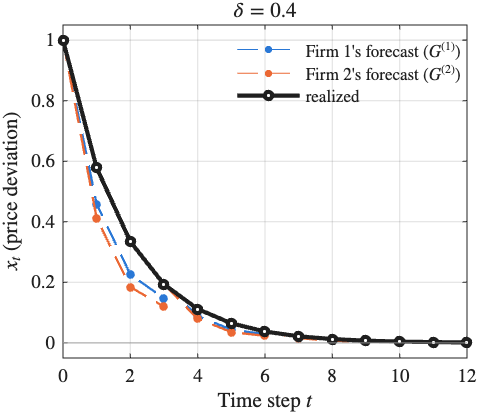}
  \end{minipage}\hfill
  \begin{minipage}[b]{0.32\textwidth}
    \centering
    \includegraphics[width=\linewidth]{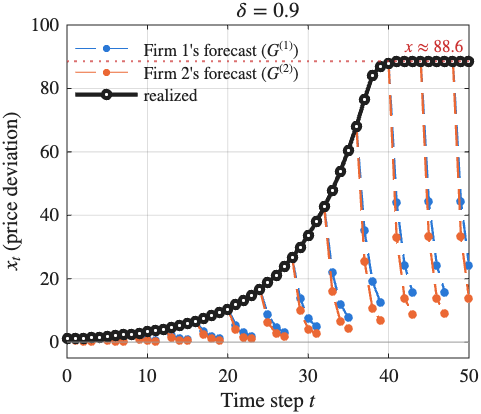}
  \end{minipage}
  \caption{Duopoly simulation results. The left plot shows the spectral radius at the origin against growing misspecification level $\delta$, for the heterogeneous loop and for each conjectured game played by both firms.
  The center ($\delta = 0.4$) and right ($\delta = 0.9$) plots show the state trajectory of the heterogeneous loop from $x_0 = 1$ (solid) and each firm's Nash forecast (dashed).}
  \label{fig:duopoly}
\end{figure*}
We consider an LQ duopoly setting to demonstrate the dynamical system behavior in feedback with heterogeneous MPG controllers. Two firms maximize profit around a steady state of the market. Following the usual LQ approximation, each firm minimizes a quadratic profit loss that penalizes deviations of the price from its steady-state level and the cost of adjusting supply, including the effect of the rival's adjustment. 
The dynamics are $x_{t+1} = 0.65\, x_t + u_{1,t} + u_{2,t}$, where state $x_t$ is the price deviation and $u_{i,t} \in [-50, 50]$ is firm~$i$'s supply adjustment, with $u_{i,t} > 0$ meaning reduced supply, which raises the price, and the horizon is $K = 3$.
The true costs are in the form of~\eqref{eq:stage_cost} with $q_i = 0$, $Q_1 = 0.6$, $Q_2 = 0.24$, $R_1 = \left[\begin{smallmatrix} 1.6 & 1.9 \\ 1.9 & 2.6 \end{smallmatrix}\right]$ and $R_2 = \left[\begin{smallmatrix} 2.1 & -0.4 \\ -0.4 & 3.0 \end{smallmatrix}\right]$.
Firm~1 cares more about the price and gains from moving against its rival ($[R_1]_{12} > 0$, business stealing), while Firm~2 adjusts slowly and tends to follow ($[R_2]_{12} < 0$).

Each firm knows its own true cost, and its conjectured cost of rival $\ell \neq j$ moves from the truth to increasingly misspecified values as $\delta \in [0,1]$ grows, $(\hat Q_\ell^{(j)}, \hat R_\ell^{(j)}) = (1-\delta)(Q_\ell, R_\ell) + \delta(\bar Q_\ell^{(j)}, \bar R_\ell^{(j)})$, with $\bar Q_2^{(1)} = 1.26$, $\bar R_2^{(1)} = \left[\begin{smallmatrix} 2.0 & -0.2 \\ -0.2 & 0.8 \end{smallmatrix}\right]$, $\bar Q_1^{(2)} = 0.72$ and $\bar R_1^{(2)} = \left[\begin{smallmatrix} 1.9 & -1.1 \\ -1.1 & 0.9 \end{smallmatrix}\right]$.
Firm~1 sees Firm~2 as a price defender that cares about the price and adjusts cheaply, whereas Firm~2 mistakes Firm~1 for a follower. Every conjectured game is strongly monotone ($\mu_j \geq 0.52$ for all $\delta \in [0,1]$), and Assumption~\ref{as:mpc} holds for every $\delta \in [0,1]$.

The left plot of Fig.~\ref{fig:duopoly} shows the change in the spectral radius at the origin for growing levels of misspecification.
If both firms use the same conjectured game, the loop is stable for every $\delta$. When each firm uses its own misspecified game model, the heterogeneous loop~\eqref{eq:closed_loop_system} has a locally asymptotically stable origin for $\delta < \delta^\star\approx 0.84$, and an unstable origin for $\delta >\delta^\star$. Theorem~\ref{thm:stability} additionally certifies global asymptotic stability for $\delta <\delta_c \approx 0.48.$ The gap between $\delta_c$ and $\delta^\star$ reflects the conservatism of \Cref{thm:stability}. 

The center plot of Fig.~\ref{fig:duopoly} shows the scenario at $\delta = 0.4$, where the certificate of Theorem~\ref{thm:stability} is satisfied and heterogeneous loop reaches the origin, which is the globally asymptotically stable equilibrium point.
The right plot of Fig.~\ref{fig:duopoly} shows the scenario at $\delta = 0.9$, where the certificate of Theorem~\ref{thm:stability} is not satisfied and the closed loop has no globally asymptotically stable equilibrium.
Firm~1 predicts that Firm~2 will correct strongly and therefore pushes the price away while Firm~2 actually corrects only weakly. Both firms forecast a decay at every step while the price grows. Since $\mathcal U_i$ is compact, once Firm~1's prediction of Firm~2 reaches the input bound, the price settles at $x \approx \pm 88.6$ (depending on the sign of $x_0$). The closed loop then has three equilibria, including the origin, and thus no unique equilibrium. These limiting prices grow with the bounds of the input sets, and increasing the bounds of $\mathcal U_i$'s moves them further from the origin.

\section{CONCLUSION}
In this paper, we studied multi-agent dynamical systems with MPG controllers involving objective misspecifications resulting in heterogeneity in agents' conjectures.
We showed that stability can be preserved in the presence of objective misspecifications with \Cref{thm:stability} under \Cref{as:mpc}. 
We further showed in a financial duopoly simulation that heterogeneity in agents' conjectured games can remove the globally asymptotically stable equilibrium that exists when all agents share the same game model.
Future work would include stability analysis with additional misspecifications in the agents' assumptions of the system model and constraint sets, to capture agents' limited knowledge of the system dynamics and each other's capabilities.



\section*{APPENDIX}
\renewcommand{\thesubsection}{\Alph{subsection}}

\subsection{Settings of Motivating Example}\label{app:counterexample}
The motivating example uses an LTI system defined with $n = 2$ agents, $d = m_1 = m_2 = 1$, dynamics as $x_{t+1} = 0.8\, x_t + u_{1,t} + u_{2,t}$, input bounds as $ u_{i,t} \in [-50, 50]$, and prediction horizon $K = 3$.
 
The stage costs of game $G^{(j)}$ have the form of~\eqref{eq:stage_cost},
with $u_t = (u_{1,t}, u_{2,t})$, $q_i^{(j)} = 0$, and the weights
$Q_1^{(1)} = 0.15$, $R_1^{(1)} = \left[\begin{smallmatrix} 3.3 & 2.9 \\ 2.9 & 3.3 \end{smallmatrix}\right]$,
$Q_2^{(1)} = 1.1$, $R_2^{(1)} = \left[\begin{smallmatrix} 5.7 & -2.2 \\ -2.2 & 2.6 \end{smallmatrix}\right]$ in $G^{(1)}$, and
$Q_1^{(2)} = 0.35$, $R_1^{(2)} = \left[\begin{smallmatrix} 1.3 & 0.1 \\ 0.1 & 1.5 \end{smallmatrix}\right]$,
$Q_2^{(2)} = 0.05$, $R_2^{(2)} = \left[\begin{smallmatrix} 3.4 & 1.6 \\ 1.6 & 1.8 \end{smallmatrix}\right]$ in $G^{(2)}$. Agent~1 plays $G^{(1)}$ and agent~2 plays
$G^{(2)}$, so each agent's own cost is correct, and its conjecture of the other agent's cost is not. 
Both games satisfy the stability certificate of~\cite{hall_stability_2025}.
 
At every time step, each agent computes the NE of its own game from the measured state $x_t$ and applies its first input $u_{i,0}$. Near the origin, no input bound is active. Each MPG controller reduces to a linear feedback $u_{i,t} = k_i x_t$ where $k_i$ is the first-step feedback gain of agent $i$ in the game it plays. The closed loop is $x_{t+1} = \rho\, x_t$ with $\rho = 0.8 + k_1 + k_2$. The origin of this scalar linear system is asymptotically stable if $|\rho| < 1$, since then $|x_{t+1}| = |\rho|\,|x_t|$ is a contraction, and unstable if $|\rho| > 1$~\cite{hespanha2018linear}. The open loop has $\rho = 0.8$. When both agents play $G^{(1)}$ or both play $G^{(2)}$, $\rho = 0.712$ or $\rho = 0.730$, respectively. When each agent plays its own game, with
$k_1 = 0.212$ and $k_2 = 0.266$, $\rho = 1.278 > 1$: the controllers turn a stable open loop into an unstable origin.

\subsection{Proof of \Cref{thm:stability}}\label{app:proof}
The proof consists of five parts, first identifying block elements of the system, then analyzing their networked structure.

\noindent\textit{1) Definition of feedback system blocks:} The feedback system blocks for the closed-loop system in \eqref{eq:closed_loop_system} are visualized in Fig.~\ref{fig:proof_feedback}. The system blocks are of two kinds: $\Psi_0$ as the LTI system in (\ref{eq:dynamics}) representing the multi-agent dynamics, and the collection of $\Psi_j$'s being the joint control action feedback composed of each agent's individual MPG controller defined by local finite horizon game solutions in the form of mappings $\Sj{j}(\cdot)$ derived from the results of each agent's solution mappings.
The properties of LTI systems are already well understood; we derive key properties of our static, non-linear feedback $\kappa$.

The expressions for the two terms in the pseudo-gradient mapping of each agent \( j \in N \) in \eqref{eq:solution_map_general} are given as:
\begin{equation*}
\begin{aligned}
\Fu{j}(u) \hspace{-2pt} &= \hspace{-2pt}
{\footnotesize
2\begin{bmatrix}
\Bt{1}^\top \Qt{1}{j} \Bt{1}+\tilde{R}_{1_{11}}^{(j)} \hspace{-10pt}& \cdots & \hspace{-10pt} \Bt{1}^\top \Qt{1}{j} \Bt{n}+\tilde{R}_{1_{1n}}^{(j)} \\
\vdots & \ddots & \vdots \\
\Bt{n}^\top \Qt{n}{j} \Bt{1}+\tilde{R}_{n_{n1}}^{(j)} \hspace{-10pt}& \cdots & \hspace{-10pt} \Bt{n}^\top \Qt{n}{j} \Bt{n}+\tilde{R}_{n_{nn}}^{(j)}
\end{bmatrix}} u \\
&\quad +
{\footnotesize
\begin{bmatrix}
(\Bt{1})^\top \qt{1}{j} \\
\vdots \\
(\Bt{n})^\top \qt{n}{j}
\end{bmatrix}},
\qquad
\Fx{j} =
{\footnotesize
\begin{bmatrix}
2(\Bt{1})^\top \Qt{1}{j} \At \\
\vdots \\
2(\Bt{n})^\top \Qt{n}{j} \At
\end{bmatrix}},
\end{aligned}
\end{equation*}
where \(\Qt{i}{j} = \blkdiag(I_K \otimes Q_i^{(j)}, 0_{d\times d})\), \(\qt{i}{j} = \col(\mathbf{1}_K \otimes q_i^{(j)}, 0_{d})\) and \(\tilde{R}_{i_{ik}}^{(j)} = I_K \otimes R_{i_{ik}}^{(j)}\), with $\otimes$ the Kronecker product, $\mathbf{1}_K$ the vector of ones, $_{ik}$ the block position within $R_i^{(j)}$, and the zero blocks reflecting that \eqref{eq:fh_game} has no terminal cost on $x_K$;
\( \At \) is the free response matrix of the global dynamics \eqref{eq:dynamics}, and \( \Bt{i} \) the impulse response matrix of agent \(i\), which are expressed explicitly as:
\begin{equation}
\At =
{\footnotesize
\begin{bmatrix}
I \\
A \\
\vdots \\
A^K
\end{bmatrix}},
\quad
\Bt{i} =
{\footnotesize
\begin{bmatrix}
0 & \cdots & \cdots & 0 \\
B_i & 0 & \cdots & 0 \\
AB_i & B_i & \cdots & 0 \\
\vdots & \vdots & \ddots & \vdots \\
A^{K-1} B_i & \cdots & AB_i & B_i
\end{bmatrix}}.
\label{eq:stacked_AB}
\end{equation}

Recall that each agent's finite horizon NE prediction is the solution to a VI.
VIs can also be expressed as normal cone inclusion problems of the form $F(u) + \NC_{\Z}(u) \ni 0.$
The solution mapping $\Sj{j}$ can be rewritten as $\Sj{j}(\xinit)\nonumber=\phij{j}(-\Fx{j}\xinit)$
by defining $\phij{j}$, a static nonlinearity, as:
$
     \phij{j}(\cdot) = \bigl( \Fu{j} + \NC_{\Z} \bigr)^{-1} (\cdot).
$

The first block, $\Psi_0$, representing the LTI system is:
\begin{equation}
\Psi_0 :
\left\{
\begin{aligned}
    x_{t+1} &= A x_t + z_t \\
    y_t &= x_{t}
\end{aligned}
\right.
\label{eq:proof_lti_block}
\end{equation}
where $z_t=\sum_{j \in N} B_j \Xi_j u_t^{(j)}$. Therefore, the input and output of system \( \Psi_0 \), are \( z_t \) and \(y_t\), respectively.
This LTI system block is in feedback with the collection of the individual static nonlinear maps $\phij{j}$, each represented as $\Psi_{j} : u_t^{(j)} = \phij{j}\big(c_t^{(j)}\big)$,
where the input and output of each \( \Psi_j \) are \(c_t^{(j)}\) and \(u_t^{(j)}\), respectively, and $c_t^{(j)} = -\Fx{j}\, y_t$.

\noindent\textit{2) Existence of a dynamical system equilibrium:} We start by considering the set of equilibrium points of the closed-loop dynamical system in \eqref{eq:closed_loop_system} given by $\Eq = \{ \xbar \mid \xbar = A\xbar + B \kappa (\xbar)\}$. The mapping 
$h(x) = (I - A)^{-1}B\kappa(x)$ is defined, and the following conditions are checked to conclude the existence of at least one fixed point. (i) $h(x)$ is single-valued: Through \Cref{as:mpc} and \Cref{prop:singleton}, the individual NE for every agent's individual conjectured game $\Gj{j}$ inside their MPG controllers exist and are unique. Therefore, the MPG controllers admit unique solutions as outputs, and $h(x)$ takes in $\kappa$, which is a linear function of these unique solutions. (ii) $h(x)$ has the same set of fixed points as the closed-loop dynamics \eqref{eq:closed_loop_system}. (iii) $h(x)$ is continuous since individual mappings $\phij{i}$'s are continuous by \cite[Proposition~2]{hall_stability_2025}, and $h(x)$ takes in $\kappa$, which is a linear function of these continuous mappings. (iv) $h$ maps $\R^{d}$ into the set $\mathcal{H} \coloneqq \big\{ (I-A)^{-1} B u \ \big|\ u \in \textstyle\prod_{i \in N} \U_i \big\}$, which is nonempty, convex, and compact: $\prod_{i \in N} \U_i$ is convex and compact under \Cref{as:mpc}, and $\mathcal{H}$ is its image under the linear map $(I-A)^{-1}B$. By (i), (iii), and (iv), $h$ is therefore a continuous, single-valued self-map of the nonempty compact convex set $\mathcal{H}$, and Brouwer's fixed-point theorem~\cite[Thm.~2.1.18]{facchinei_finite-dimensional_2004} yields at least one fixed point $\xbar \in \mathcal{H}$ satisfying $\xbar = h(\xbar)$. By (ii), $\xbar \in \Eq$, i.e., $\xbar$ is an equilibrium point of \eqref{eq:closed_loop_system}.

\begin{figure}
    \centering
    \includegraphics[width=\linewidth]{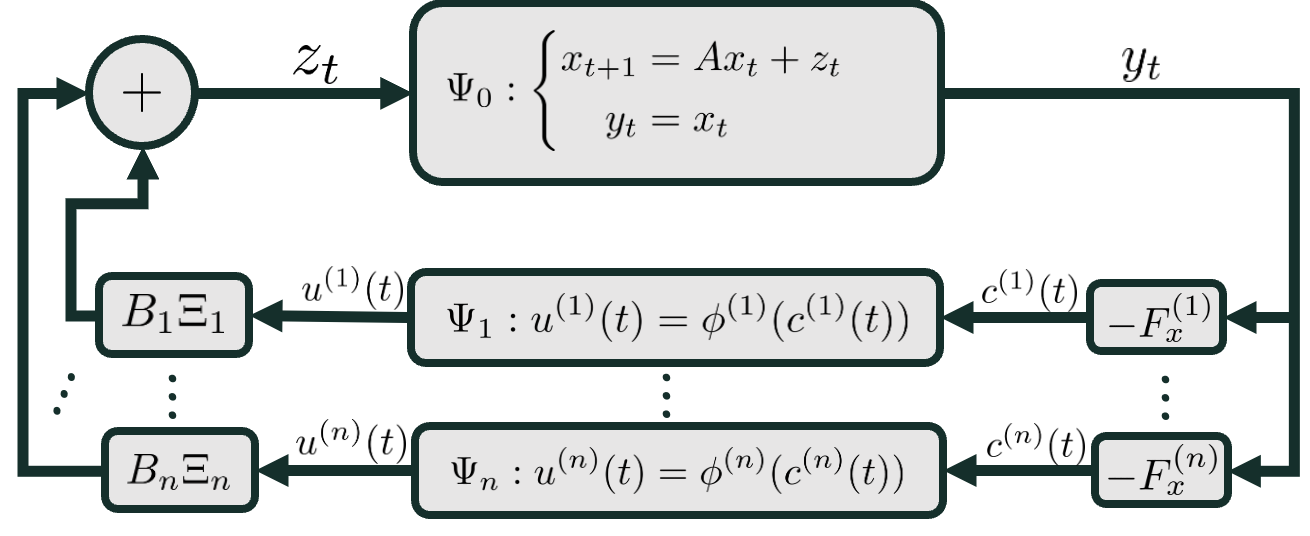}
    \caption{Block diagram of the feedback connection between $\Psi_0$ and $\Psi_j$'s, modeling heterogeneous MPG system.}
    \label{fig:proof_feedback}
\end{figure}

\noindent\textit{3) Inequalities for dissipativity and input-output relations:} Dissipation inequalities with storage functions are derived for each main subsystem block, namely, $V_0$ and $V_j$'s for each subsystem $\Psi_0$ and $\Psi_j$'s separately.
The following are defined: $\Delta x_t = (x_t - \xbar)$, $\Delta z_t = (z_t - \bar{z})$, $\Delta y_{t} = (y_{t} - \bar{y})$, $\Delta u^{(j)}_{t} = (u^{(j)}_{t} - \bar{u}^{(j)})$, and $\Delta c^{(j)}_{t} = (c^{(j)}_{t} - \bar{c}^{(j)}) $. As $\Psi_0$ is a discrete-time LTI system, a quadratic storage function family is chosen as $V_{\xbar}^0(x_t) = (\Delta x_t)^\top P (\Delta x_t)$ for some positive definite $P$. Evaluating the evolution of $V_{\xbar}^0$ through the time steps results in:
\begin{equation}
V_{\xbar}^0(x_{t+1}) \hspace{-1pt}-\hspace{-1pt} V_{\xbar}^0(x_t) = 
\begin{bmatrix}
\Delta y_{t} \\
\Delta z_{t}
\end{bmatrix}^{\top}
\begin{bmatrix}
A^\top P A - P & A^\top P  \\
\star & P 
\end{bmatrix}
\begin{bmatrix}
\Delta y_{t} \\
\Delta z_{t}
\end{bmatrix}
\label{eq:proof_storage_lti}
\end{equation}

$\phij{j}(\cdot)$'s are static nonlinear mappings, which are memoryless and, thus, do not store energy. Therefore, their storage functions $V_{\xbar}^{j}$'s are equal to 0 \cite{haddad_stability_2019}. By \cite[Proposition~2]{hall_stability_2025}, $\phij{j}(\cdot)$ are $\monc_j$-cocoercive, therefore, the following relationship between the input and the output is known, and achieved by rearranging the  $\monc_j$-cocoercivity inequality:
\begin{equation} \label{eq:mapping_cocoerc}
(\Delta u_t^{(j)})^\top \Delta c_t^{(j)}-\monc_j (\Delta u_t^{(j)})^\top \Delta u_t^{(j)}\;\ge\; 0
\end{equation}
Since $V_{\xbar}^{j}=0$, \eqref{eq:mapping_cocoerc} and $V_{\xbar}^{j}$ are compared to each other to achieve the following relationship between the storage function $V_{\xbar}^{j}$ and the input-output of $\Psi_j$:
\begin{equation}
V_{\xbar}^{j}\leq 
\begin{bmatrix}
\Delta u^{(j)}_{t} \\
\Delta c^{(j)}_{t}
\end{bmatrix}^{\top}
\begin{bmatrix}
-\monc_j I & \frac{1}{2}I \\
\star & 0
\end{bmatrix}
\begin{bmatrix}
\Delta u^{(j)}_{t} \\
\Delta c^{(j)}_{t}
\end{bmatrix},
\label{eq:proof_storage_nonlinear}
\end{equation}

\noindent\textit{4) Connection of main blocks $\Psi_0$ and $\Psi_j$'s:} Next, the connection between the two main dissipativity blocks are applied through the interconnection equations given as
$\Delta z_t = \sum_{j \in N}B_j\Xi_j\Delta u^{(j)}_{t}$ and $ \Delta c^{(j)}_{t} =-\Fx{j}\Delta y_t.$

By substituting the interconnection equations into the storage function inequalities and adding them altogether, each inequality \eqref{eq:proof_storage_nonlinear} weighted by its own multiplier $\lambda_j>0$, we can write the overall inequality for the interconnection of $\Psi_0$ and $\Psi_j$'s in terms of $\Delta y_{t}$ and  $\Delta u^{(j)}_{t}$'s collected under 
\(\Delta \hat{u}_t= \col(\Delta u^{(j)}_{t})_{j \in N} \).
This is represented as $V_{\xbar}(x_{t+1}) -V_{\xbar}(x_t) \hspace{-1pt} =V_{\xbar}^{0}(x_{t+1}) - V_{\xbar}^{0}(x_t) + \sum_{j\in N}\lambda_j V_{\xbar}^{j}$,
which corresponds to:
\begin{equation} \label{eq:thm_derivation}
{\footnotesize\begin{aligned}
&V_{\xbar}(x_{t+1}) \hspace{-2pt} - \hspace{-2pt} V_{\xbar}(x_t)\le
\begin{bmatrix}
\Delta y_{t} \\
\Delta \hat{u}_t
\end{bmatrix}^{\top} \hspace{-5pt}
\Bigl(M(P) + \sum_{j\in N}\lambda_j W_j\Bigr)
\begin{bmatrix}
\Delta y_{t} \\
\Delta \hat{u}_t
\end{bmatrix}
\end{aligned}}
\end{equation}
where $\lambda_j>0$ for $j\in N$, $M(P)$ is defined in \eqref{eq:M_def}, $\Bhat =\begin{bmatrix} {B}_1\Xi_1 & {B}_2\Xi_2 & \dots & {B}_n\Xi_n \end{bmatrix}$, and the blocks of each $W_j$ are defined as in \eqref{eq:W_def}.
Thus, if
\begin{equation}
\begin{aligned}
&\begin{bmatrix}
A^\top P A - P & A^\top P \Bhat \\
\Bhat^\top P A & \Bhat^\top P \Bhat
\end{bmatrix}
+ \sum_{j\in N}\lambda_j W_j
\preceq -\varepsilon I
\end{aligned}
\label{eq:proof_lmi}
\end{equation}
for some $\varepsilon > 0$, then $V_{\xbar}(x_{t+1}) - V_{\xbar}(x_t) < 0$ for all $x_t \neq \xbar$.

\noindent\textit{5) Stability of the overall system:}  The discrete-time Lyapunov theorem \cite{haddad_stability_2019} is applied to show the closed-loop stability of the overall system. Conditions (i) $V_{\xbar}(\xbar) = 0$, (ii) $V_{\xbar}(x) > 0$, $\forall x \in \R^{d} \setminus \{\xbar\}$, and (iii) $V_{\xbar}(x) \to \infty$ as $\|x\| \to \infty$ are satisfied as a quadratic storage function is selected. Condition (iv) $V_{\xbar}(Ax + B\kappa(x)) - V_{\xbar}(x) < 0$, $\forall x \in \R^{d} \setminus \{\xbar\}$ is satisfied as the decrease condition is imposed by \eqref{eq:proof_lmi}. The equilibrium point $\xbar$ is globally asymptotically stable for $\varepsilon > 0$, and hence the unique equilibrium of \eqref{eq:closed_loop_system}.
\hfill $\blacksquare$

\section*{ACKNOWLEDGMENT}

We used large language model tools, including Claude Code with Claude Opus~5.5 (Anthropic), to assist with software implementation, proofreading, and review. All final text, figures, code, and references were edited, revised, and reviewed by humans.

\bibliography{references,references-2}
\bibliographystyle{IEEEtran}

\end{document}